\documentclass[runningheads]{llncs}

\usepackage{cite}
\usepackage{amsmath,amssymb,amsfonts}
\usepackage{mathdots}
\usepackage{algorithmicx}
\usepackage{graphicx}
\usepackage{caption}
\usepackage{subcaption}
\usepackage{textcomp}
\usepackage[dvipsnames]{xcolor}
\usepackage{orcidlink}
\usepackage{hyperref}
\usepackage{todonotes}
\usepackage{pgfplots}
\usepackage{pgf}
\usepackage{graphicx}
\usepackage{svg}
\usepackage{multirow}
\usepackage{bbm}
\usepackage{tikz}
\usetikzlibrary{quantikz2}

\usepackage[english]{babel}
\usepackage{algorithm2e}
\RestyleAlgo{ruled}
\usepackage[noend]{algpseudocode}
\usepackage{dsfont}
\usepackage{listings}

\usepackage[capitalise]{cleveref}

\usepackage[T1]{fontenc}
\usepackage{graphicx}
\pgfplotsset{compat=1.18}
\begin{document}
\title{Dead Gate Elimination}
%
%
\author{
    Yanbin Chen\orcidID{0000-0002-1123-1432} \and
    Christian B.~Mendl\orcidID{0000-0002-6386-0230} \and
    Helmut Seidl\orcidID{0000-0002-2135-1593}
}
\authorrunning{Y. Chen, et al.}
%
\institute{
School of CIT, Technical University of Munich, Garching 85748, Germany
\email{\{yanbin.chen, christian.mendl, helmut.seidl\}@tum.de}
}
\maketitle              
\begin{abstract}
Hybrid quantum algorithms combine the strengths of quantum and classical computing. Many quantum algorithms, such as the variational quantum eigensolver (VQE), leverage this synergy. However, quantum circuits are executed in full, even when only subsets of measurement outcomes contribute to subsequent classical computations. 
In this manuscript, we propose a novel circuit optimization technique that identifies and removes dead gates. We prove that the removal of dead gates has no influence on the probability distribution of the measurement outcomes that contribute to the subsequent calculation result. 
We implemented and evaluated our optimization on a VQE instance, a quantum phase estimation (QPE) instance, and hybrid programs embedded with random circuits of varying circuit width, confirming its capability to remove a non-trivial number of dead gates in real-world algorithms. The effect of our optimization scales up as more measurement outcomes are identified as non-contributory, resulting in a proportionally greater reduction of dead gates.
\end{abstract}
\section{Introduction}\label{sec:intro}
In recent efforts to address complex real-world problems, researchers are increasingly integrating quantum and classical computing to use the unique strengths of both paradigms \cite{8638598}. In such interdisciplinary development, domain specialists are exploring ways to implement or even accelerate specific subroutines through quantum circuits tailored to quantum processing units (\emph{QPU}s) \cite{bouland2020prospectschallengesquantumfinance, 10313603, 10313887, quetschlich2024equivalencecheckingclassicalcircuits, jojo2024quantumalgorithmstensorsvd}. Concurrently, quantum experts may incorporate classical computing procedures, given the wealth of sophisticated classical computing procedures that have been developing over decades \cite{10313756, doi:10.1137/S0097539795293172, De_Luca_2021}. 
A popular algorithm framework that allows to take advantage of the strength of both quantum and classical computers is \emph{hybrid programs} \cite{mccaskey2018hybrid}.
In hybrid programs, quantum circuits are embedded as subroutines into programs from a classical host language.
Usually, the classical host program handles optimization, control, and data processing, while the quantum circuits are used for specific calculations that may benefit from quantum speedups.

However, this integration can present challenges\cite{elsharkawy2023integrationquantumacceleratorshigh, rohe2024problemsolutiongeneralpipeline, 10313875}. When researchers work beyond their core expertise, the interplay between classical and quantum components may be suboptimal. This imperfect coupling risks inefficient resource utilization.

An implicit assumption is often made that the circuits are executed as external entities and all qubits are measured in the end and their outcomes are collected for purposes that are not of interest to circuits.
So, circuits are fully executed even if not all measurement outcomes contribute to later calculations.
However, in the following example, we see the potential for circuit simplification when knowing that some measurement outcomes are not needed.

\begin{example}
    In the hybrid program in \cref{fig:example-hybrid-program-1st-measurement-outcome-not-used}, the measurement outcome $o_0$ does not contribute to final results: in $\textbf{Proc}_a$, the initial value of variable $a$, i.e. $o_0$, gets canceled out in the expression $z - 2t$; in $\textbf{Proc}_b$ the initial value of variable $a$ has no impact on the return value, because $0 \leq \eta a \leq 0.5$ and thus the $\eta a$  part is always rounded down to $0$ by $\textbf{int}(\cdot)$ operator. So, if we execute the circuit \cref{fig:circuit-simple-observation}, where the measurement outcome from $q_0$ is always discarded, and we assign an arbitrary value from $\{0,1\}$ to $o_0$, the results of the both $\textbf{Proc}_a$ and $\textbf{Proc}_b$ will not be influenced. 
    Then, we could optimize the program by running the simplified circuit \cref{fig:circuit-simple-observation-after-removal} instead of circuit \cref{fig:circuit-simple-observation}.
    We call gates removed by this analysis dead gates. We will formally justify that this simplification will never influence calculation results of hybrid programs in \cref{sec:methods}.
\end{example}

\begin{figure}[htb]
\centering
\begin{tikzpicture}[>=latex,node distance=1.5em]
 \node(a)
 {
 \begin{quantikz}[row sep=0.1cm]
        \lstick{$q_0$}& \gate[3]{U_3}&\targ{}   &\gate{V_1}&\gate{V_2}&\meter{}\rstick{$o_0 \in \{0, 1\}$}\\
        \lstick{$q_1$}&              &\ctrl{-1} &\ctrl{-1} &          &\meter{}\rstick{$o_1 \in \{0, 1\}$} \\
        \lstick{$q_2$}&              &\gate{W_1}&\ctrl{-2} &\ctrl{-2} &\meter{}\rstick{$o_2 \in \{0, 1\}$}
\end{quantikz}
 };
 \node[below=of a](b)
 {
    \colorbox{pink}{
        \begin{minipage}{0.3\textwidth}
            \SetAlgoLined
            \SetNlSty{textbf}{}{:}
            \begin{algorithmic}
            \State $a, b, c \gets o_0, o_1, o_2$
            \State $x \gets 2a - 3b$
            \State $y \gets b + c + 2$
            \State $z, t \gets xy, ay$
            \State \Return $z - 2t$
            \end{algorithmic}
      \end{minipage}
    }
 };
 \node[below=of b](c)
 {
    \colorbox[rgb]{0.74,0.83,1}{
        \begin{minipage}{0.3\textwidth}
            \SetAlgoLined
            \SetNlSty{textbf}{}{:}
            \begin{algorithmic}
            \State $a, b, c \gets o_0, o_1, o_2$
            \State $\eta \gets \textbf{random}(0.1, 0.5)$
            \State \Return $\textbf{int}(\eta a + bc)$
            \end{algorithmic}
      \end{minipage}
    }
 };
 \node[draw] at (-4.5,0.5) {$\textbf{QC}$};
 \node[draw] at (-4,-2.6) {$\textbf{Proc}_a$};
 \node[draw] at (-4,-5.5) {$\textbf{Proc}_b$};
 \path[->]
(a) edge [] node [red, near start] {$o_0, o_1, o_2\qquad\qquad \ $} (b)
(a) edge [bend left=65] node [blue, below, near end] {$\qquad\qquad o_0, o_1, o_2 $} (c);
\end{tikzpicture}
\caption{An example of a hybrid program, where a quantum circuit $\textbf{QC}$ of $3$ qubits are first executed and then the measurement outcomes $o_0$, $o_1$, $o_2$ from qubits $q_0$, $q_1$, $q_2$, respectively, are dispatched to one of the two classical computing procedures, $\textbf{Proc}_a$, or  $\textbf{Proc}_b$.}
\label{fig:example-hybrid-program-1st-measurement-outcome-not-used}
\end{figure}
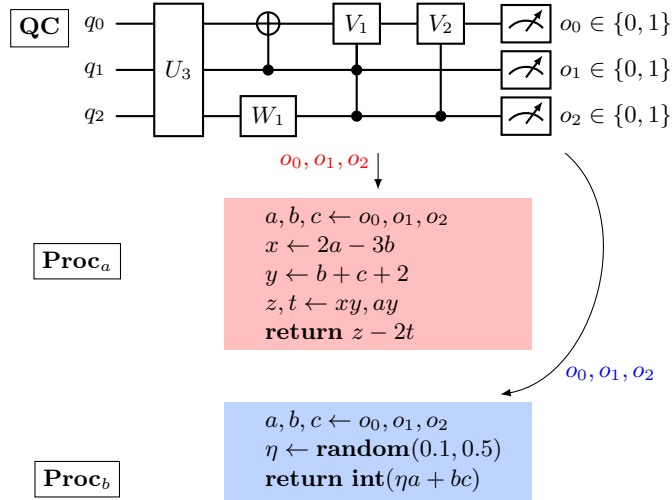

\begin{figure}
    \centering
    \begin{quantikz}[row sep=0.08cm, column sep=0.33cm]
        \lstick{$q_0$}& \gate[3]{U_3}&\targ{}   &\gate{V_1}&\gate{V_2}&\ground{}\\
        \lstick{$q_1$}&              &\ctrl{-1} &\ctrl{-1} &          &\meter{} \\
        \lstick{$q_2$}&              &\gate{W_1}&\ctrl{-2} &\ctrl{-2} &\meter{}
    \end{quantikz}
    \caption{A $3$-qubit circuit. The measurement outcome of the top qubit is discarded.}
    \label{fig:circuit-simple-observation}
\end{figure}
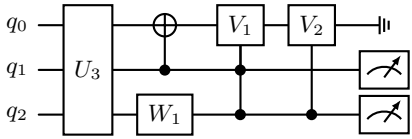

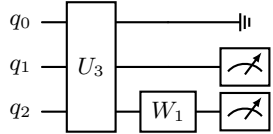
\begin{figure}
    \centering
    \begin{quantikz}[row sep=0.08cm, column sep=0.33cm]
        \lstick{$q_0$}& \gate[3]{U_3}&          &\ground{}\\
        \lstick{$q_1$}&              &          &\meter{} \\
        \lstick{$q_2$}&              &\gate{W_1}&\meter{}
    \end{quantikz}
    \caption{Simplification of the circuit in \cref{fig:circuit-simple-observation}. The probability distribution on measurement outcomes that are not discarded remains unchanged.}
    \label{fig:circuit-simple-observation-after-removal}
\end{figure}

Non-contributory measurement outcomes can also occur in scenarios where the classical computing procedure only queries a subset of the available measurement results. This could make the same simplification possible, as we will demonstrate in the VQE and QPE examples in \cref{sec:evaluation}.

In addtion, qubits that are not explicitly measured, such as ancilla qubits used in intermediate computations can also be interpreted in the same manner. While not explicitly measured at the end of circuits, these qubits could be considered as implicitly measured, with their outcomes being discarded immediately. In this perspective, such qubits fit naturally into the consideration of our paper, as their measurement outcomes do not influence subsequent classical computations.

Several existing works address certain aspects of the matter discussed in this paper.
The partial equivalence checking proposed in \cite{chen_2022}, verifies whether two circuits yield the same probability distribution for a given set of measurement outcomes, but it does not provide a way to simplify circuits while preserving these distributions. Moreover, it requires explicit global unitary operators, the computation of which is infeasible for large-scale circuits due to the inherent complexity of circuit simulation. QuTracer proposed in \cite{li2024qutracermitigatingquantumgate} optimizes circuits by eliminating gates that do not affect a subset of measured qubits, but it lacks a formal framework for this process and may fail to recognize redundant operations—such as SWAP gates that merely permute qubits without altering measurement distributions. In \cite{Abedi_2023}, it is mentioned that a measurement outcome depends only on its causal light cone, yet it does not provide a systematic method to exploit this insight for circuit simplification. Crucially, existing approaches overlook a key optimization opportunity: in hybrid quantum-classical workflows, some measurement outcomes become non-contributory to subsequent classical computations.

In this manuscript, we introduce a novel approach to simplify circuits that uses context information from the classical computing components of the hybrid program. By propagating the contextual information that some measurement outcomes are not contributory, our method identifies and removes dead gates in circuits without changing the semantics of the entire hybrid program, leading to more resource-efficient circuits and quantum-classical integration. We evaluate our method by running it on instances of VQE and QPE algorithm, and on random circuits in \cref{sec:evaluation}.

\section{Preliminaries}\label{sec:prelim}
This manuscript assumes that readers are familiar with the basics of quantum computing. For a detailed introduction to quantum computing, we recommend the following literature \cite{nielsen_QC_2012, rieffel2000introduction, kaye2006introduction}.
In this section, we explain some notations that we will use later.

For an $n$-qubit circuit $C$, we use $C.\textbf{gates}\text{()}$ to denote the set of all gates in $C$. Each gate in $C.\textbf{gates}\text{()}$ is an object storing information, including gate type, the set of qubits it acts on, and the set of gates it depends on. We use $C$ to also denote the unitary matrix of circuit $C$ if no ambiguity is produced. For example, when applying $C$ to a $n$-qubit state $S$, the resulting state is $CS$. We use the following $-$ as an operator to remove one gate from a circuit.

\begin{definition}[$-$ operator]
    For an $n$-qubit circuit $C$ and a gate $g \in C.\textbf{gates}\text{()}$, $C - g$ represents a $n$-qubit circuit obtained by removing the gate $g$ from $C$, namely $(C-g).\textbf{gates}\text{()} :=  C.\textbf{gates}\text{()}\backslash \{g\}$.
\label{def:minus}
\end{definition}

We use the following notation to describe the probability of the measurement outcomes of a subsystem of a quantum state being a given binary string. 
\begin{definition}[Probability distribution of subsystem measurement outcomes]
    For an $n$-qubit state $S_n$ on a set of qubits $Q_n = \{q_0, \dots, q_{n-1}\}$, and a binary string $k$ of length $|k| = |Q|$, where $Q$ is a subset of qubits $\{q_{i_0}, \dots, q_{i_{|k|-1}}\}= Q \subseteq Q_n$ where $i_0< \dots < i_{|k|-1}$,  $\mathcal{P}^{k}_{i_0\dots i_{|k|-1}}[S_n]$ denotes the probability of $q_{i_j}$ measuring $k[j]$ for all $j\in \{0, \dots, |k|-1\}$, where $k[j]$ is the $j$-th element of $k$.
\label{def:prob-dist}
\end{definition}

\begin{example}
    Consider a $2$-qubit state $\ket{\Phi}$ on qubits $q_0$ and $q_1$, where $\ket{\Phi} = \alpha_0\ket{00} + \alpha_1\ket{01} + \alpha_2\ket{10}+ \alpha_3\ket{11}$.
    $\mathcal{P}^{01}_{01}[\ket{\Phi}]$ represents the probability of measuring $0$ on $q_0$ and $1$ on $q_1$, namely the probability of the state collapsing to $\ket{01}$, which is $|\alpha_1|^2$. Similarly, $\mathcal{P}^{10}_{01}[\ket{\Phi}] = |\alpha_2|^2$, $\mathcal{P}^{00}_{01}[\ket{\Phi}] = |\alpha_0|^2$. $\mathcal{P}^{1}_{0}[\ket{\Phi}]$ represents the probability of measuring $1$ on $q_0$, which is $|\alpha_2|^2 + |\alpha_3|^2$, because when $q_0$ is measured $1$, $q_1$ could be measured either $0$ or $1$.
\end{example}

\begin{definition}[Frontier]
    Given a circuit $C$, its \emph{frontier} is a set  $\mathcal{F}_C$ satisfying:
    (a) $\mathcal{F}_C \subseteq C.\textbf{gates}\text{()}$; 
    (b) for any gate $g\in\mathcal{F}_C$, any output wire of $g$ is no input of any other gates.
\label{def:frontier}
\end{definition}

\begin{example}
    The frontier of the following circuit only consists of $V_5$ and $U_2$.
    
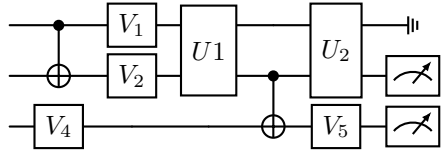
\begin{equation*}
        \begin{quantikz}[row sep=0.1cm, column sep=0.33cm]
            &\ctrl{1}  &\gate{V_1}&\gate[2]{U1}&          &\gate[2]{U_2}   &\ground{}\\
            &\targ{}   &\gate{V_2}&            &\ctrl{1}  &          &\meter{}\\
            &\gate{V_4}&          &            &\targ{}   &\gate{V_5}    &\meter{} 
        \end{quantikz}
    \end{equation*}
\end{example}

\section{Method}\label{sec:methods}
In this work, We restrict our discussion to circuits that contain no mid-circuit measurements or resets.
We start by introducing some concepts that we will use in later discussions.

For a quantum circuit $C$, we assume that for the outcomes we collect by measuring all qubits, a subset of them has no contribution to the classical computing procedures that come later. We explicitly mark such measurement outcomes as discarded, and we call them \emph{discarded measurement outcomes}. The following notation is put at the end of a qubit wire to denote that the measurement outcome on that qubit is discarded:
    \begin{quantikz}
        &\ground{}
    \end{quantikz}.
On the contrary, a \emph{valid measurement outcome} is the one that is not discarded.

From now on, if a measurement outcome is valid, we omit the symbol of measurement at the end of the qubit wire in the circuit diagram for conciseness.

\begin{definition}[Dead/Valid qubit]
    A qubit is a \emph{dead qubit} if its measurement outcome is discarded.
    A qubit is a \emph{valid qubit} if it is not dead.
\end{definition}

Next, we establish an equivalent relation among circuits that is based on measurements performed only on valid qubits. That is, in this equivalence, we consider two circuits to be equal if their probability distributions of valid measurement outcomes are identical.

\begin{definition}
[Equivalence relative to valid outcomes]
Given two circuits $C_1$ and $C_2$ applied on the same set of qubits $Q_n = \{q_0, \dots, q_{n-1}\}$, for a subset $D\subseteq Q_n$ where all qubits in $D$ are dead, $C_1$ and $C_2$ are equivalent relative to $D$, denoted by $C_1 \equiv_D C_2$, if and only if for any $n$-qubit state $S$ and any binary string $k$ of length $|k| = |Q_n\backslash D|$, $\mathcal{P}^{k}_{i_0\dots i_{|k|-1}}[C_1S] = \mathcal{P}^{k}_{i_0\dots i_{|k|-1}}[C_2S]$, where $Q_n\backslash D = \{q_{i_0},\dots, q_{i_{|k|-1}} \}$ and $i_0 < \dots < i_{|k|-1}$.

\begin{example}
    The following two circuits are equivalent relative to their valid outcomes, because the probability of measuring $0$ on the valid qubit, $q_1$, is the same in both circuits.

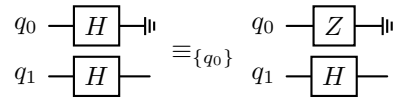
\begin{equation*}
    \begin{quantikz}[row sep=0.15cm, column sep=0.33cm]
        \lstick{$q_0$}&\gate{H}&\ground{}\\
        \lstick{$q_1$}&\gate{H}&
    \end{quantikz}
    \equiv_{\{q_0\}}
    \begin{quantikz}[row sep=0.15cm, column sep=0.33cm]
        \lstick{$q_0$}&\gate{Z}&\ground{}\\
        \lstick{$q_1$}&\gate{H}&
    \end{quantikz}
\end{equation*}

\end{example}

\label{def:eq-valid-outcomes}
\end{definition}

Then, we move on to concepts of dead gates, which are essential to our method. Given the knowledge that some measurement outcomes do not influence subsequent calculations and we discard them explicitly, we define a gate as dead if removing it only affects the probability distribution of these discarded measurement outcomes.

\begin{definition}[Dead gate]
    Given a circuit $C$ and a gate $g$ in $C$, and a set of dead qubits $D$, $g$ is a \emph{dead gate} if and only if $C \equiv_D C^{\prime} := C - g$, where $-$ is defined by \cref{def:minus}.
    \label{def:dead-gate}
\end{definition}

Since removing dead gates does not change the probability distribution on valid measurement outcomes, we could simplify circuits by removing such dead gates.
By \cref{theorem:remove-single-qubit-gates-on-unused-wire}, \cref{theorem:remove-multi-ctrl-gates-acting-on-unused-wire}, and \cref{theorem:remove-swap}, we present our approach to identify dead gates and prove that removing these dead gates does not influence the results of calculation, therefore justifying the correctness of our method.

\begin{theorem}
    Given any operator $U$ acting on $n + 1$ qubits and any operator $V$ acting on a single qubit $q_i$, and $q_i$ is dead, it holds that
        \begin{equation}
            \begin{quantikz}[row sep=0.05cm, column sep=0.33cm]
           \lstick{$q_i$} & & \gate[2]{U}&\gate{V}& \ground{}\\
            & \qwbundle{n}          &               & & \qw
            \end{quantikz}
            \equiv_{\{q_i\}}
            \begin{quantikz}[row sep=0.05cm, column sep=0.33cm]
           \lstick{$q_i$} && \gate[2]{U}& \ground{}\\
            & \qwbundle{n}          &&      \qw
            \end{quantikz}\,.
            \label{equation:theorem-remove-single-qubit-gates-on-unused-wire}
        \end{equation}\
    \label{theorem:remove-single-qubit-gates-on-unused-wire}
\end{theorem}
\begin{proof}
This is a special case of \cref{theorem:remove-multi-ctrl-gates-acting-on-unused-wire}. $\square$

\end{proof}

\begin{remark}
By \cref{def:dead-gate}, gate $V$ in \cref{equation:theorem-remove-single-qubit-gates-on-unused-wire} is a dead gate, so we could optimize the circuit by removing it.
\end{remark}

\begin{theorem}
     Given any operator $U$ acting on $n + 1$ qubits and any operator $V$ acting on a single qubit $q_i$, and $V$ is controlled by $n_c$ qubits, where $n_c + n_r = n$, and $q_i$ is dead, it holds that
        \begin{equation}
            \begin{quantikz}[row sep=0.1cm]
           \lstick{$q_i$} && \gate[3]{U} &\gate{V}& \ground{}\\
            & \qwbundle{n_c}                        &&\ctrl{-1} & \qw \\
            & \qwbundle{n_r}&&&
            \end{quantikz}
            \equiv_{\{q_i\}}
            \begin{quantikz}[row sep=0.1cm]
           \lstick{$q_i$} && \gate[3]{U} & \ground{}\\
            & \qwbundle{n_c}                        & & \qw \\
            & \qwbundle{n_r}&&
            \end{quantikz}\,.
            \label{equation:theorem-remove-multi-ctrl-gates-acting-on-unused-wire}
        \end{equation}\
    \label{theorem:remove-multi-ctrl-gates-acting-on-unused-wire}
\end{theorem}
\begin{proof}
    Let the circuit on the left be $C_1$, and the circuit on the right be $C_2$.
    W.l.o.g, we assume that $i = 0$, and the gate $V$ is controlled by qubits $q_1, \dots, q_{n_c}$.
    Suppose the gate $V$ is defined by $V\ket{0} = \alpha_{v_0}\ket{0} + \beta_{v_0}\ket{1}$ and $V\ket{1} = \alpha_{v_1}\ket{0} + \beta_{v_1}\ket{1}$.
    For any input state $S$, we assume that
        $\ket{\Phi} = C_2S =\sum^{N - 1}_{j = 0} c_j\ket{j}$, 
    where $N = 2^{n+1}$, $c_j \in \mathbb{C}$.
    Then, for any $n$-bit binary string $k$, we have
        $\mathcal{P}^{k}_{1...n}[\ket{\Phi}] = |c_{0 \oplus k}|^2 + |c_{1 \oplus k}|^2$, 
    where $\oplus$ is string concatenation (E.g., $00\oplus 11 = 0011$ and $110 \oplus 1 = 1101$).
    In fact, $\ket{\Phi}$ could be rewritten as
    \begin{align}
    \begin{split}
        &\ket{\Phi} = \sum^{N - 1}_{j = 0} c_j\ket{j} = \sum_{\substack{|s| = n_c, \\0\in s}}\sum_{|t| = n_r} c_{0\oplus s\oplus t}\ket{0\oplus s\oplus t} + \\
        &\sum_{\substack{|s| = n_c, \\0\notin s}}\sum_{|t| = n_r} c_{0\oplus s\oplus t}\ket{0\oplus s\oplus t} + 
        \sum_{\substack{|s| = n_c, \\0\in s}}\sum_{|t| = n_r} c_{1\oplus s\oplus t}\ket{1\oplus s\oplus t} + \\
        &\sum_{\substack{|s| = n_c, \\0\notin s}}\sum_{|t| = n_r} c_{1\oplus s\oplus t}\ket{1\oplus s\oplus t}        
    \end{split}
    \end{align}
    So, by applying $C_1$ to $S$, the output state $\ket{\Psi} = C_1 S=(C^{n_c}V)C_2S$ is
    \begin{align}
    \begin{split}
        &\ket{\Psi} = C^{n_c}V \ \ket{\Phi} = 
        \sum_{\substack{|s| = n_c, \\0\notin s}}\sum_{|t| = n_r} \sum^{1}_{b=0} c_{b\oplus s\oplus t}\alpha_{v_b}\ket{0\oplus s\oplus t} + \\
        &\sum_{\substack{|s| = n_c, \\0\notin s}}\sum_{|t| = n_r} \sum^{1}_{b=0} c_{b\oplus s\oplus t}\beta_{v_b}\ket{1\oplus s\oplus t} + 
        \sum_{\substack{|s| = n_c, \\0\in s}}\sum_{|t| = n_r}\sum^{1}_{b=0} c_{b\oplus s\oplus t}\ket{b\oplus s\oplus t}
    \end{split}
    \end{align}
    where $C^{n_c}V$ denotes the multi-controlled gate $V$.\\
    If $\exists l\in \{1,\dots,n_c\}$: $k[l] = 0$, then
    $\mathcal{P}^{k}_{1...n}[C_1S]$ is calculated by
    \begin{align}
        \sum^{1}_{b=0}\sum_{s\oplus t = k} |c_{b\oplus s\oplus t}|^2
        = |c_{0\oplus k}|^2+
        |c_{1\oplus k}|^2
        = \mathcal{P}^{k}_{1...n}[C_2S]
    \end{align}    
    If $\forall l \in \{1,\dots,n_c\}$: $k[l] = 1$, then  $\mathcal{P}^{k}_{1...n}[C_1S]$ is calculated by
    \begin{align}
    \begin{split}
      &\sum^{1}_{b=0}\sum_{s\oplus t = k} |c_{b\oplus s\oplus t}\alpha_{v_b}|^2+|c_{b\oplus s\oplus t}\beta_{v_b}|^2 
      = \sum^{1}_{b=0}\sum_{s\oplus t = k} |c_{b\oplus s\oplus t}|^2 (|\alpha_{v_b}|^2+|\beta_{v_b}|^2) \\
      &= |c_{0\oplus k}|^2(|\alpha_{v_0}|^2+|\beta_{v_0}|^2) +
      |c_{1\oplus k}|^2(|\alpha_{v_1}|^2+|\beta_{v_1}|^2)
      = \mathcal{P}^{k}_{1...n}[C_2S]  
    \end{split}
    \end{align}
    Since our choice of $S$ is arbitrary, by \cref{def:eq-valid-outcomes}, \cref{equation:theorem-remove-multi-ctrl-gates-acting-on-unused-wire} holds. $\square$
\end{proof}

It could happen that removing some gate makes some dead qubits valid and some valid qubits dead, while the probability distribution of valid measurement outcomes is unchanged. For our analysis to encompass this case, we need to extend the equivalence in \cref{def:eq-valid-outcomes} and the dead gate in \cref{def:dead-gate}.

\begin{definition}[Extended equivalence relative to valid outcomes]
Given circuits $C_1$ and $C_2$ applying on the set of qubits $Q_n = \{q_0, \dots, q_{n-1}\}$, for $D1, D2 \subseteq Q_n$ where $|D_1| = |D_2|$ and all qubits in $D1$ and $D_2$ are dead, $C_1 \equiv^{D_1}_{D_2} C_2$ 
iff for any $n$-qubit state $S$ and any binary string $k$ of length $|k| = |Q_n\backslash D_1| = |Q_n\backslash D_2|$, $\mathcal{P}^{k}_{i_0\dots i_{|k|-1}}[C_1S] = \mathcal{P}^{k}_{[e_1/f_1,\dots,e_m/f_m](i_0\dots i_{|k|-1})}[C_2S]$, where $Q_n\backslash D_1 = \{q_{i_0},\dots, q_{i_{|k|-1}} \}$, $i_0 < \dots < i_{|k|-1}$, $D_1\backslash(D_1\cap D_2) = \{e_1,\dots,e_m\}$, $D_2\backslash(D_1\cap D_2) = \{f_1,\dots,f_m\}$, and $[b_1/a_1,\dots,b_p/a_p]s$ denotes a string obtained by for each $l \in \{1,\dots,p\}$ replacing $a_l$ in string $s$ with $b_l$ (E.g., $[1/4,2/5,3/6]456=123$).
\label{def:ext-eq-valid-outcome}
\end{definition}

\begin{definition}[Extended dead gate]
    Given a circuit $C$ and a gate $g$ in $C$ acting on a set of qubits $Q$, $g$ is a dead gate if and only if there exist subsets $D_1, D_2 \subseteq Q$ such that $C \equiv^{D_1}_{D_2} C^{\prime} = C - g$, where $-$ is defined by \cref{def:minus}.
    \label{def:ext-dead-gate}
\end{definition}

\begin{theorem}
     Given any operator $U$ acting on $n + 2$ qubits and a SWAP gate, it holds that
        \begin{equation}
            \begin{quantikz}[row sep=0.1cm]
           \lstick{$q_i$} && \gate[3]{U} &\swap{1}& \ground{}\\
           \lstick{$q_j$} &                        &&\targX{}& \qw \\
            & \qwbundle{n}&&&
            \end{quantikz}
            \equiv^{\{q_i\}}_{\{q_j\}}
            \begin{quantikz}[row sep=0.07cm]
           \lstick{$q_i$} && \gate[3]{U} & \\
           \lstick{$q_j$} &                        && \ground{} \\
            & \qwbundle{n}&&
            \end{quantikz}\,.
            \label{equation:theorem-remove-swap}
        \end{equation}\
    \label{theorem:remove-swap}
\end{theorem}
\begin{proof}
    It follows directly the definition of SWAP gates and \cref{def:ext-eq-valid-outcome}. $\square$
\end{proof}
\begin{remark}
    The SWAP gate in \cref{equation:theorem-remove-swap} is a dead gate by \cref{def:ext-dead-gate} and can be removed. After removing a SWAP gate, we also need to adapt the qubit mapping, if the qubit mapping/routing is performed at an earlier stage. 
\end{remark}

Our optimization algorithm is shown in \cref{alg:optimization}, of which the asympotic bound is given in \cref{theorem:alg-complexity}.
\begin{algorithm}
\caption{Dead gates removal}
\label{alg:optimization}
\KwData{$C \in circuits$}
\KwResult{$C_{opt}$}
$C_{opt} \gets C$, 
$terminate \gets \textbf{False}$\;
\While{$terminate \neq True\land \emptyset \neq \mathcal{F_C} \gets C_{\text{opt}}.\textbf{frontier}\text{()}$}{
    $terminate \gets \textbf{True}$\;
    \For{$g \in \mathcal{F_C}$}{
        \If{$g$ is a dead gate by \cref{theorem:remove-single-qubit-gates-on-unused-wire}, \cref{theorem:remove-multi-ctrl-gates-acting-on-unused-wire}, or \cref{theorem:remove-swap}}{
            $C_{opt} \gets C_{opt} - g$,
            $terminate \gets \textbf{False}$\;
        }
    }
}
\end{algorithm}
\begin{theorem}[\cref{alg:optimization} is polynomial]
    The time complexity of \cref{alg:optimization} is  $\mathcal{O}(|C.\textbf{gates}\text{()}|^2)$.
\label{theorem:alg-complexity}
\end{theorem}
\begin{proof}
    In each $\textbf{while}$ iteration, $\mathcal{O}(|C.\textbf{gates}\text{()}|)$-many gates are checked to see whether they are dead, and since at least one gate is removed in every iteration (except the last iteration), there are at most $\mathcal{O}(|C.\textbf{gates}\text{()}|)$-many iterations. $\square$
\end{proof}

One should be cautious when considering circuit simplification based on the knowledge about non-contributory measurement outcomes. There are cases where it seems that some gate only "writes" on dead qubits and looks like a dead gate, but it is not a dead gate and thus cannot be removed, as demonstrated in the following example.

\begin{example}
    The following simplification would, in general, lead to a changed probability distribution on valid measurement outcomes. 
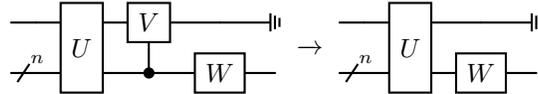
\begin{equation*}
        \begin{quantikz}[row sep=0.15cm, column sep=0.33cm]
             &     &         \gate[2]{U} &\gate{V} &&\ground{} \\
            & \qwbundle{n} & &\ctrl{-1}& \gate{W} &
        \end{quantikz}
        \xrightarrow{}
        \begin{quantikz}[row sep=0.15cm, column sep=0.33cm]
                 & &\gate[2]{U} &&\ground{} \\
            & \qwbundle{n}& & \gate{W}& 
        \end{quantikz}
\end{equation*}
    To see this, we consider a special case of it shown as follows:
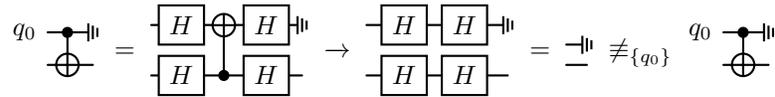
\begin{equation*}
\begin{quantikz}[row sep=0.15cm, column sep=0.1cm]
    \lstick{$q_0$}&\ctrl{1}&\ground{} \\
    &\targ{} &
\end{quantikz}
=
\begin{quantikz}[row sep=0.15cm, column sep=0.1cm]
     &\gate{H}&\targ{}&\gate{H}&\ground{} \\
    &\gate{H}&\ctrl{-1}&\gate{H} &
\end{quantikz}
\xrightarrow{}
\begin{quantikz}[row sep=0.15cm, column sep=0.2cm]
     &\gate{H}&\gate{H}&\ground{} \\
    &\gate{H}&\gate{H} &
\end{quantikz}
=
\begin{quantikz}[row sep=0.15cm, column sep=0.2cm]
 &\ground{} \\
 &
\end{quantikz}
\not\equiv_{\{q_0\}}
\begin{quantikz}[row sep=0.15cm, column sep=0.1cm]
    \lstick{$q_0$}&\ctrl{1}&\ground{} \\
    &\targ{} &
\end{quantikz}
\end{equation*}
\end{example}

\section{Evaluation}\label{sec:evaluation}
We conduct $3$ sets of experiments to evaluate our method. In the first set, our method is applied to the quantum algorithm VQE. In the second set, an instance of QPE is optimized with our method. In the third set, our method is applied to random circuits. The demo implementation of optimization and experiments is accessible at \url{https://github.com/i2-tum/demo-dead-gate-elimination}.

\paragraph{VQE algorithm}
The VQE is a hybrid quantum-classical algorithm that finds the ground state energy of a quantum system, and it is widely used in areas like quantum chemistry and material science \cite{fedorov2022vqe}.
In each iteration of the VQE algorithm, an \emph{Ansatz}, a parameterized circuit selected from a diverse range of designs, is executed, and measurements are performed on all output qubits. These measurement outcomes are then scheduled to an \emph{optimizer}, a classical computing procedure. This procedure uses them to calculate expectation values of Hamiltonian terms, which are then used to update parameters in the Ansatz.

Due to the broad range of applications of VQE, it has been integrated into well-established toolchains such as Qiskit, allowing it to be utilized as a black-box subroutine \cite{Qiskit}.
While this facilitates the use of quantum computers for domain experts, it also introduces the risk of misalignment between quantum circuits and classical computing procedures,
particularly because there are already numerous choices of Ansatz with different focuses, and many more are expected to be developed in the future\cite{peruzzo2014variational, McClean_2016, Romero_2019}.

Consider the instance of the VQE algorithm constructed in \cref{fig:example-vqe-alg}. In each iteration, the 4-qubit Ansatz $A_1$ is executed and measured. The resulting measurement outcomes, $o_i$ ($i \in {0,\dots,3}$), are then sent to the optimizer. There, two expectation values, $\mathbb{E}_Z$ and $\mathbb{E}_X$, are computed and combined. The resulting values are used to update the parameters in the Ansatz—namely, $\overrightarrow{\theta_r}$ and $\theta_j$ ($j \in {1,\dots,8}$)—which are adapted for the next iteration.

Here, we observe that the calculation of $\mathbb{E}_Z$ and $\mathbb{E}_X$ only depends on the measurement outcomes $o_2$ and $o_3$, meaning $q_0$ and $q_1$ are dead qubits. By applying \cref{alg:optimization} to the Ansatz $A_1$, a simplified Ansatz $A_2$ shown in \cref{fig:simplified-vqe-ansatz} is obtained. Thus, we can optimize the VQE instance by replacing 
$A_1$ with $A_2$, reducing the number of parameterized gates by $4$ and the number of two-qubit gates by $3$ in each iteration. Considering that VQE requires many iterations to converge, the reduction in gate operations becomes even more significant.

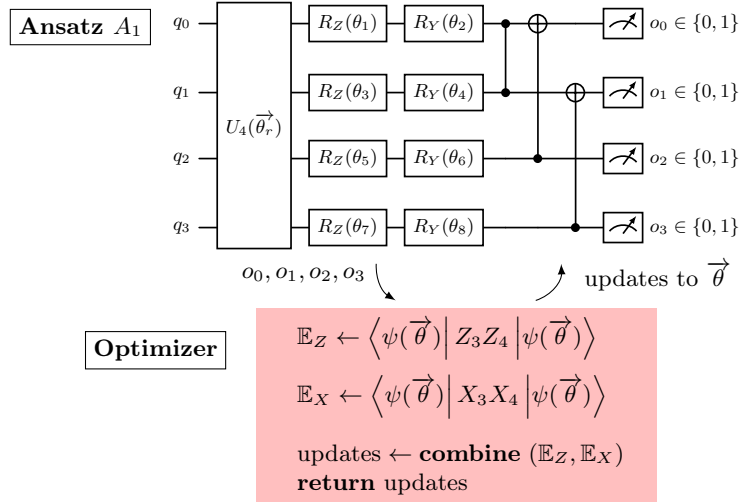
\begin{figure}[htb]
\centering
\begin{tikzpicture}[>=latex,node distance=1.5em]
 \node(a)
 {
\scalebox{.8}{
\begin{quantikz}[column sep = 0.3cm]
        \lstick{$q_0$}&\gate[4]{U_4(\overrightarrow{\theta_r})}&\gate{R_Z(\theta_1)}&\gate{R_Y(\theta_2)}&\ctrl{1}   &\targ{}  &         &\meter{}\rstick{$o_0 \in \{0, 1\}$}\\
        \lstick{$q_1$}&           &\gate{R_Z(\theta_3)}&\gate{R_Y(\theta_4)}&\control{} &         &\targ{}  &\meter{}\rstick{$o_1 \in \{0, 1\}$}\\
        \lstick{$q_2$}&           &\gate{R_Z(\theta_5)}&\gate{R_Y(\theta_6)}&           &\ctrl{-2}&         &\meter{}\rstick{$o_2 \in \{0, 1\}$}\\
        \lstick{$q_3$}&           &\gate{R_Z(\theta_7)}&\gate{R_Y(\theta_8)}&           &         &\ctrl{-2}&\meter{}\rstick{$o_3 \in \{0, 1\}$}  
\end{quantikz}
}

 };
 \node[below=of a](b)
 {
    \colorbox{pink}{
        \begin{minipage}{0.4\textwidth}
            \SetAlgoLined
            \SetNlSty{textbf}{}{:}
            \begin{algorithmic}
            \State $\mathbb{E}_Z \gets\bra{\psi(\overrightarrow{\theta})}Z_3Z_4\ket{\psi(\overrightarrow{\theta})}$
            \State 
            \State $\mathbb{E}_X \gets\bra{\psi(\overrightarrow{\theta})}X_3X_4\ket{\psi(\overrightarrow{\theta})}$
            \State
            \State $\text{updates} \gets \textbf{combine}\text{ (}\mathbb{E}_Z, \mathbb{E}_X\text{)}$
            \State \Return $\text{updates}$
            \end{algorithmic}
      \end{minipage}
    }
 };
 \node[draw] at (-5,1.3) {$\textbf{Ansatz }A_1$};
 \node[draw] at (-4,-3) {$\textbf{Optimizer}$};
 \path[->]
(a) edge [bend right=30] node [near start] {$o_0, o_1, o_2, o_3\qquad\qquad\qquad $} (b)
(b) edge [bend right=37] node [near end] {$\qquad\qquad\qquad\qquad \text{updates to }\overrightarrow{\theta}$} (a);
\end{tikzpicture}
\caption{An instance of VQE algorithm.}
\label{fig:example-vqe-alg}
\end{figure}

\begin{figure}
    \centering
    \begin{quantikz}[row sep=0.15cm, column sep=0.33cm]
        \lstick{$q_0$}&\gate[4]{U_4(\overrightarrow{\gamma})}&                   &&\\
        \lstick{$q_1$}&           &&&\\
        \lstick{$q_2$}&           &\gate{R_Z(\theta_5)}&\gate{R_Y(\theta_6)}&\\
        \lstick{$q_3$}&           &\gate{R_Z(\theta_7)}&\gate{R_Y(\theta_8)}&  
    \end{quantikz}
    \caption{A simplified Ansatz, $A_2$, that can replace $A_1$ in \cref{fig:example-vqe-alg}.}
    \label{fig:simplified-vqe-ansatz}
\end{figure}
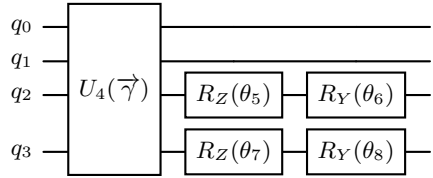
\paragraph{QPE algorithm}
The QPE algorithm determines the phase associated with an eigenvalue of a given unitary operator \cite{kitaev1995quantum, dobvsivcek2008quantum, nielsen_QC_2012}. QPE is employed as a subroutine in various quantum algorithms, among which one of the most famous examples is Shor's algorithm \cite{doi:10.1137/S0097539795293172}.

Consider an instance of QPE constructed in \cref{fig:instance-QPE}, where the QPE circuit is executed and its measurement outcomes constitute the estimated phase $\theta \in [0,1)$ received by the classical computing procedure $\textbf{Proc}_c$. However, we can observe that the most significant bit of $\theta$, namely $\theta_0$, is always subtracted away in the expression $\lambda - \lfloor \lambda \rfloor$. Hence, the initial value of $\theta_0$, i.e., $o_0$, is not contributory, and we identify $q_0$ as a dead qubit. Then, by running our optimization on the QPE circuit, we remove $m$ two-qubit gates and a Hadamard gate and get a simplified circuit in \cref{fig:QPE-optimized-circuit}. Thus we can optimize the QPE instance by replacing the circuit in \cref{fig:instance-QPE} by the circuit \cref{fig:QPE-optimized-circuit}.

\begin{figure}[htb]
\centering
\begin{tikzpicture}[>=latex,node distance=1.5em]
 \node(a)
 {
 \begin{quantikz}[row sep=0.1cm, column sep=0.33cm]
        \lstick{$q_0$}     &\gate[5]{U}&\gate[4]{\text{QFT}^{\dagger}_f}&\gate{R_m^{\dagger}}& \ \ldots &\gate{R_1^{\dagger}}&\gate{H}     &\meter{}\rstick{$o_0$}\\
        \lstick{$q_1$} &&&& \ \ldots& \ctrl{-1} &       &\meter{}\rstick{$o_1$} \\ \\
        \lstick{$q_{m}$} &&&\ctrl{-3}& \ \ldots &&&\meter{}\rstick{$o_m$}\\
        \lstick{\ket{\psi}}&&\qwbundle{r}&&\ \ldots&&&
\end{quantikz}
 };
 \node[below=of a](b)
 {
     \colorbox[rgb]{0.74,0.83,1}{
        \begin{minipage}{0.4\textwidth}
            \SetAlgoLined
            \SetNlSty{textbf}{}{:}
            \begin{algorithmic}
            \State $\theta_0, \dots \theta_m \gets o_0, \dots, o_m$
            \State $\theta \gets 0.\theta_0\dots\theta_m$
            \State $\lambda \gets 10\theta$
            \State \Return $\lambda - \lfloor \lambda \rfloor$
            \end{algorithmic}
      \end{minipage}
    }
 };
 \node[draw] at (-4.7,1.0) {$\textbf{QPE}$};
 \node[draw] at (-4,-2.5) {$\textbf{Proc}_c$};
 \node[] at (-4,0) {$\vdots$};
 \node[] at (0.5,0) {$\iddots$};
 \node[] at (4,0) {$\vdots$};
 \path[->]
(a) edge [] node [blue, near end] {$\qquad\qquad\quad \ o_0, \dots, o_m$} (b);
\end{tikzpicture}
\caption{An instance of QPE, where $U$ consists of Hadamard gates and a sequence of controlled oracles to prepare the state ready for the inverse quantum Fourier transform (QFT \cite{coppersmith2002approximatefouriertransformuseful, nielsen_QC_2012}), $\text{QFT}^{\dagger}_f$ is the front part of the inverse QFT, measurement outcomes $o_i \in \{0, 1\}$ for all $i$, and $\lfloor \cdot \rfloor$ is the floor function that maps a real value to the greatest integer less than or equal to it.}
\label{fig:instance-QPE}
\end{figure}
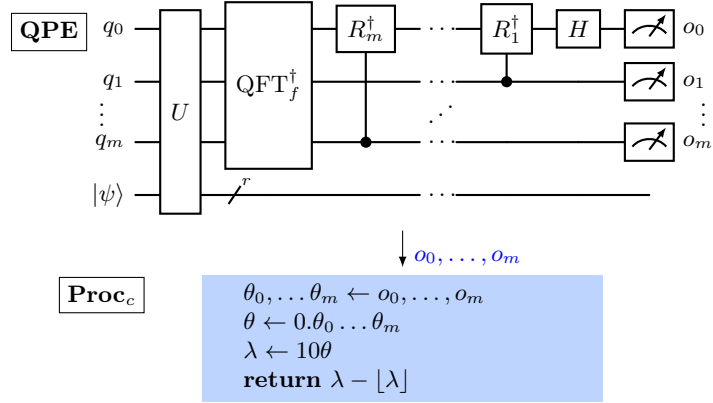

\begin{figure}[htb]
\centering
\begin{tikzpicture}[>=latex,node distance=1.5em]
 \node(a)
 {
 \begin{quantikz}[row sep=0.05cm, column sep=0.51cm]
        \lstick{$q_0$}     &\gate[5]{U}&\gate[4]{\text{QFT}^{\dagger}_f}*&\meter{}\rstick{$o_0$}\\
        \lstick{$q_1$}  &&&\meter{}\rstick{$o_1$} \\ \\
        \lstick{$q_{m}$} &&&\meter{}\rstick{$o_m$}\\
        \lstick{\ket{\psi}}&\qwbundle{r}&&
\end{quantikz}
 };

 \node[] at (-2,0) {$\vdots$};
 \node[] at (2.2,0) {$\vdots$};
\end{tikzpicture}
\caption{A simplified QPE circuit to replace the circuit in \cref{fig:instance-QPE}.}
\label{fig:QPE-optimized-circuit}
\end{figure}
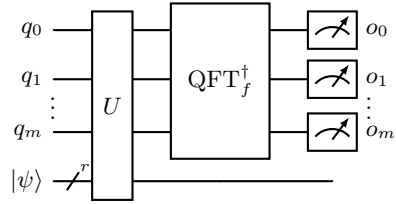

\paragraph{Random circuits}
As an effort to ensure a broad and unbiased evaluation of our optimization algorithm, we also conduct experiments where \cref{alg:optimization} is performed on randomly generated circuits.
We consider hybrid programs consisting of alternating quantum and classical segments, illustrated as follows:
\begin{equation*}
    \begin{tikzpicture}
    \node[fill=blue!60](QC0)
     {
    $\text{QC}_0$
     };
     \node[right=of QC0, fill=red!60](CC0)
     {
    $\text{C}_0$
     };
     
    \node[right=of CC0](ddd)
     {
    $\dots$
     };
     \node[right=of ddd, fill=blue!60](QCi)
     {
    $\text{QC}_{b-1}$
     };
     \node[right=of QCi, fill=red!60](CCi)
     {
    $\text{C}_{b-1}$
     };
     \draw[->] (QC0) -- (CC0);
     \draw[->] (CC0) -- (ddd);
     \draw[->] (ddd) -- (QCi);
     \draw[->] (QCi) -- (CCi);
     
    \end{tikzpicture}
\end{equation*}
Each hybrid program consists of a sequence of $b$ quantum-classical blocks, where each block comprises a quantum circuit $\text{QC}_i$ followed by a classical computation $\text{C}_i$. The parameter $b$ controls the number of such blocks in the hybrid program.
In our evaluation, we set $b=60$ and generate random circuits of \emph{circuit width} ranging from $2$ to $100$ qubits, where the circuit width is defined as the number of qubits in the circuit. For each circuit, the total gate count is $100$ times the circuit width. The circuits are constructed using the universal Clifford + T gate set with single-qubit gates comprising $10\%$ of all generated gates, ensuring a reasonable balance between single- and multi-qubit operations. 
For each circuit width $w$, we generate $1000$ hybrid programs. In each program, every quantum block $\text{QC}_i$ is instantiated with a random circuit generated as described above.
All gates are placed uniformly at random across the circuit, ensuring no positional bias in their distribution.
These circuits are assumed to be first pre-optimized using circuit transpilers, such as Qiskit and t$\ket{\text{ket}}$, ensuring that the input circuits are highly optimized by the state-of-the-art compilation toolchain.

We then apply our optimization algorithm to every quantum circuit $\text{QC}_i$ within each hybrid program under various settings of dead qubit constraints. Specifically, we consider five settings: $1$, $2$, or $3$ dead qubits, as well as when the number of dead qubits is set to $10\%$ or $20\%$ of the circuit width.
For each circuit width and each dead qubit setting, we evaluate the gate count reduction achieved by our algorithm across all circuits in all generated hybrid programs. The mean gate reduction serves as our primary performance metric, providing a robust and comprehensive assessment of the algorithm's effectiveness in realistic hybrid execution scenarios.

The result of our experiments is shown in \cref{fig:result-random}. The experimental results show that our method consistently removes a non-trivial number of gates across settings. This is particularly notable given that our optimization is applied to circuits that are assumed to have already been pre-optimized using circuit transpilers. This demonstrates that our approach achieves further optimization beyond what is achievable with current state-of-the-art quantum compilation tools. 

For the settings with a fixed number of dead qubits (1, 2, and 3), we observe that the number of removed gates is initially high when the circuit width is small, and then decreases and stabilizes as the circuit width increases. This trend reflects a key observation: in small circuits, a fixed number of dead qubits represents a large proportion of the total qubit count (e.g., 1 dead qubit out of 2 or 3 is $50\%$–$33\%$), which creates substantial optimization opportunities. As the circuit grows wider, however, the proportion of dead qubits diminishes, leading to less pronounced impact from the optimization.

The early-stage behavior of settings with a fixed number of dead qubits directly parallels the trends observed in the percentage-based dead qubit settings ($10\%$ and $20\%$). In these settings, the number of removed gates grows stepwise with the circuit width, as the absolute number of dead qubits increases discretely with circuit size. These steps correspond to the increase in dead qubit count, and each jump leads to a corresponding spike in optimization gain. This confirms that our optimization method's effectiveness is primarily driven by the proportion of dead qubits rather than their absolute number alone.

To assess the practical efficiency of our method, we evaluate its runtime as a function of circuit width, as shown in \cref{fig:result-random}. The results indicate that the execution time increases approximately linearly with the circuit width, even in cases with a significant proportion of dead qubits. This empirical observation suggests that our approach remains efficient in practice. While the asymptotic complexity analysis in \cref{theorem:alg-complexity} establishes a worst-case quadratic dependence on the number of gates, our experimental results demonstrate that, for realistic circuits, the algorithm exhibits near-linear scaling with circuit width. This suggests that our method is practically efficient and scalable.

\begin{figure}[htb]
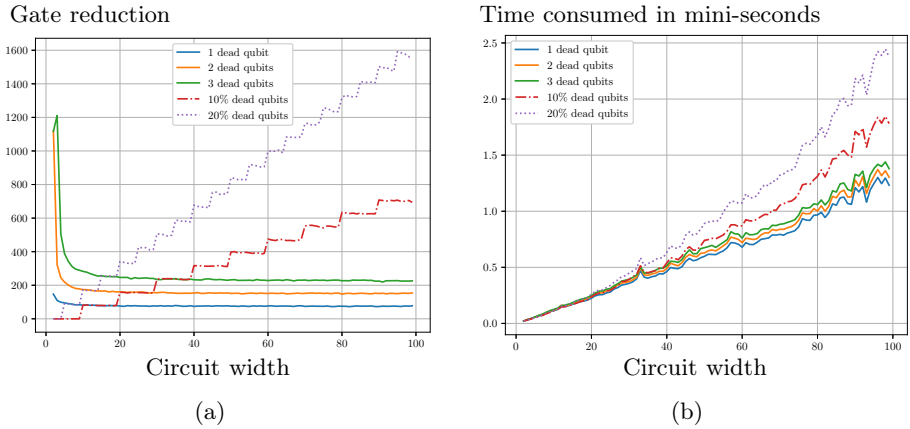

\centering
    \begin{subfigure}[b]{0.48\textwidth}
        \centering
        \begin{tikzpicture}[>=latex,node distance=0.3em]
         \node(a)
         {
        \resizebox{1\textwidth}{!}{\input{figures/100_random_results_tight_after_debug.pgf}}
         };
         \node[] at (0,-2.3)
        {\text{Circuit width}};         \node[] at (-1.7,2.4)
        {\text{Gate reduction}};
        \end{tikzpicture}
        \caption{}
        \label{fig:ext_num}
    \end{subfigure}
    \hfill
    \begin{subfigure}[b]{0.48\textwidth}
        \centering
         \begin{tikzpicture}[>=latex,node distance=0.3em]
         \node(a)
         {
        \resizebox{1\textwidth}{!}{\input{figures/time_100_random_results_tight.pgf}}
         };
         \node[] at (0,-2.3)
        {\text{Circuit width}};         \node[] at (-0.6,2.4)
        {\text{Time consumed in mini-seconds}};
        \end{tikzpicture}
        \caption{}
        \label{fig:ext_extr}
    \end{subfigure}
    \caption{(a) Gate reduction obtained by performing our optimization \cref{alg:optimization} with different dead-qubit settings on hybrid programs embedded with random circuits of circuit width ranging from $2$ to $100$. (b) The corresponding time consumed in mini seconds.}
    \label{fig:result-random}
\end{figure}

\section{Related works}\label{sec:rel_works}The concept of dead gates is inspired by the concept of dead variables in liveness analysis in compiler constructions of classical programming languages \cite{muchnick1997advanced, seidl2012compiler}.
Static analysis stands among the most potential approaches to automate the detection of dead qubits from the context of circuits, and they have proven to be effective in bug detection, program analysis, and circuit optimization \cite{paltenghi2024analyzing, zhao2023qchecker, 10189136, chen_QCP_2023, chen2024reducingmidcircuitmeasurementsprobabilistic}.

\section{Conclusion and future works}\label{sec:conclusion}
Our method demonstrates a practical approach to optimizing quantum circuits in hybrid programs by taking into account context from the classical host. By identifying dead gates and simplifying circuits accordingly, we achieve significant reductions in gate count when quantum-classical integration is suboptimal.
The evaluation of our method on the VQE and QPE instances and on random circuits confirms its potential to improve the quality of circuits. 
It would be an interesting future work to investigate how to construct the dead gate analysis on dynamic circuits. We think one of the challenges there is that mid-circuit measurements come with side effect even if they are performed on dead qubits.




\begin{credits}
\subsubsection{\ackname} This work is supported by the Bavarian state government through the project Munich Quantum Valley with funds from the Hightech Agenda Bayern Plus.

\subsubsection{\discintname}
The authors have no competing interests to declare that are relevant to the content of this article. 
\end{credits}
%
%
%
\bibliographystyle{splncs04}

\begin{thebibliography}{10}
\providecommand{\url}[1]{\texttt{#1}}
\providecommand{\urlprefix}{URL }
\providecommand{\doi}[1]{https://doi.org/#1}

\bibitem{Abedi_2023}
Abedi, E., Beigi, S., Taghavi, L.: Quantum lazy training. Quantum  \textbf{7},
  ~989 (Apr 2023). \doi{10.22331/q-2023-04-27-989},
  \url{http://dx.doi.org/10.22331/q-2023-04-27-989}

\bibitem{10313887}
Bermot, E., Zoufal, C., Grossi, M., Schuhmacher, J., Tacchino, F., Vallecorsa,
  S., Tavernelli, I.: { Quantum Generative Adversarial Networks For Anomaly
  Detection In High Energy Physics }. In: 2023 IEEE International Conference on
  Quantum Computing and Engineering (QCE). pp. 331--341. IEEE Computer Society,
  Los Alamitos, CA, USA (Sep 2023). \doi{10.1109/QCE57702.2023.00045},
  \url{https://doi.ieeecomputersociety.org/10.1109/QCE57702.2023.00045}

\bibitem{bouland2020prospectschallengesquantumfinance}
Bouland, A., van Dam, W., Joorati, H., Kerenidis, I., Prakash, A.: Prospects
  and challenges of quantum finance (2020),
  \url{https://arxiv.org/abs/2011.06492}

\bibitem{chen_2022}
Chen, T.F., Jiang, J.H.R., Hsieh, M.H.: Partial equivalence checking of quantum
  circuits. In: 2022 IEEE International Conference on Quantum Computing and
  Engineering (QCE). p. 594–604. IEEE (Sep 2022).
  \doi{10.1109/qce53715.2022.00082},
  \url{http://dx.doi.org/10.1109/QCE53715.2022.00082}

\bibitem{chen2024reducingmidcircuitmeasurementsprobabilistic}
Chen, Y., Fulginiti, I., Mendl, C.B.: {Probabilistic Circuit Model}. In: {2024
  International Conference on Quantum Computing and Engineering} (9 2024).
  \doi{10.1109/QCE60285.2024.10379}

\bibitem{chen_QCP_2023}
Chen, Y., Stade, Y.: Quantum constant propagation. In: Hermenegildo, M.V.,
  Morales, J.F. (eds.) Static Analysis. pp. 164--189. Springer Nature
  Switzerland, Cham (2023),
  \url{https://link.springer.com/chapter/10.1007/978-3-031-44245-2_9}

\bibitem{coppersmith2002approximatefouriertransformuseful}
Coppersmith, D.: An approximate fourier transform useful in quantum factoring
  (2002), \url{https://arxiv.org/abs/quant-ph/0201067}

\bibitem{De_Luca_2021}
De~Luca, G.: A survey of nisq era hybrid quantum-classical machine learning
  research. Journal of Artificial Intelligence and Technology  \textbf{2}(1),
  9–15 (Dec 2021). \doi{10.37965/jait.2021.12002},
  \url{https://ojs.istp-press.com/jait/article/view/60}

\bibitem{dobvsivcek2008quantum}
Dob{\v{s}}{\'\i}{\v{c}}ek, M.: Quantum computing, phase estimation and
  applications. arXiv preprint arXiv:0803.0909  (2008)

\bibitem{elsharkawy2023integrationquantumacceleratorshigh}
Elsharkawy, A., To, X.T.M., Seitz, P., Chen, Y., Stade, Y., Geiger, M., Huang,
  Q., Guo, X., Ansari, M.A., Mendl, C.B., Kranzlmüller, D., Schulz, M.:
  Integration of quantum accelerators with high performance computing -- a
  review of quantum programming tools (2023),
  \url{https://arxiv.org/abs/2309.06167}

\bibitem{10313875}
Elsharkawy, A., To, X.T.M., Seitz, P., Chen, Y., Stade, Y., Geiger, M., Huang,
  Q., Guo, X., Ansari, M.A., Ruefenacht, M., Schulz, L., Karlsson, S., Mendl,
  C.B., Kranzlmüller, D., Schulz, M.: Challenges in hpcqc integration. In:
  2023 IEEE International Conference on Quantum Computing and Engineering
  (QCE). vol.~02, pp. 405--406 (2023). \doi{10.1109/QCE57702.2023.10304}

\bibitem{fedorov2022vqe}
Fedorov, D.A., Peng, B., Govind, N., Alexeev, Y.: Vqe method: a short survey
  and recent developments. Materials Theory  \textbf{6}(1), ~2 (2022)

\bibitem{jojo2024quantumalgorithmstensorsvd}
Jojo, J., Khandelwal, A., Chandra, M.G.: Quantum algorithms for tensor-svd
  (2024), \url{https://arxiv.org/abs/2405.19485}

\bibitem{kaye2006introduction}
Kaye, P., Laflamme, R., Mosca, M.: An introduction to quantum computing. OUP
  Oxford (2006)

\bibitem{kitaev1995quantum}
Kitaev, A.Y.: Quantum measurements and the abelian stabilizer problem. arXiv
  preprint quant-ph/9511026  (1995)

\bibitem{li2024qutracermitigatingquantumgate}
Li, P., Liu, J., Gonzales, A., Saleem, Z.H., Zhou, H., Hovland, P.: Qutracer:
  Mitigating quantum gate and measurement errors by tracing subsets of qubits
  (2024), \url{https://arxiv.org/abs/2404.19712}

\bibitem{10313603}
Matondo-Mvula, N., Elleithy, K.: Advances in quantum medical image analysis
  using machine learning: Current status and future directions. In: 2023 IEEE
  International Conference on Quantum Computing and Engineering (QCE). vol.~01,
  pp. 367--377 (2023). \doi{10.1109/QCE57702.2023.00049}

\bibitem{8638598}
McCaskey, A., Dumitrescu, E., Liakh, D., Humble, T.: Hybrid programming for
  near-term quantum computing systems. In: 2018 IEEE International Conference
  on Rebooting Computing (ICRC). pp. 1--12 (2018).
  \doi{10.1109/ICRC.2018.8638598}

\bibitem{mccaskey2018hybrid}
McCaskey, A., Dumitrescu, E., Liakh, D., Humble, T.: Hybrid programming for
  near-term quantum computing systems. In: 2018 IEEE international conference
  on rebooting computing (ICRC). pp. 1--12. IEEE (2018)

\bibitem{McClean_2016}
McClean, J.R., Romero, J., Babbush, R., Aspuru-Guzik, A.: The theory of
  variational hybrid quantum-classical algorithms. New Journal of Physics
  \textbf{18}(2),  023023 (feb 2016). \doi{10.1088/1367-2630/18/2/023023},
  \url{https://dx.doi.org/10.1088/1367-2630/18/2/023023}

\bibitem{muchnick1997advanced}
Muchnick, S.: Advanced compiler design implementation. Morgan kaufmann (1997)

\bibitem{nielsen_QC_2012}
Nielsen, M.A., Chuang, I.L.: Quantum {{Computation}} and {{Quantum
  Information}}: 10th {{Anniversary Edition}}. {Cambridge University Press}, 1
  edn. (Jun 2012). \doi{10.1017/CBO9780511976667}

\bibitem{paltenghi2024analyzing}
Paltenghi, M., Pradel, M.: Analyzing quantum programs with lintq: A static
  analysis framework for qiskit. Proceedings of the ACM on Software Engineering
   \textbf{1}(FSE),  2144--2166 (2024)

\bibitem{peruzzo2014variational}
Peruzzo, A., McClean, J., Shadbolt, P., Yung, M.H., Zhou, X.Q., Love, P.J.,
  Aspuru-Guzik, A., O’brien, J.L.: A variational eigenvalue solver on a
  photonic quantum processor. Nature communications  \textbf{5}(1), ~4213
  (2014)

\bibitem{Qiskit}
{Qiskit contributors}: Qiskit: An open-source framework for quantum computing
  (2023). \doi{10.5281/zenodo.2573505}

\bibitem{quetschlich2024equivalencecheckingclassicalcircuits}
Quetschlich, N., Forster, T., Osterwind, A., Helms, D., Wille, R.: Towards
  equivalence checking of classical circuits using quantum computing (2024),
  \url{https://arxiv.org/abs/2408.14539}

\bibitem{rieffel2000introduction}
Rieffel, E., Polak, W.: An introduction to quantum computing for
  non-physicists. ACM Computing Surveys (CSUR)  \textbf{32}(3),  300--335
  (2000)

\bibitem{rohe2024problemsolutiongeneralpipeline}
Rohe, T., Grätz, S., Kölle, M., Zielinski, S., Stein, J., Linnhoff-Popien,
  C.: From problem to solution: A general pipeline to solve optimisation
  problems on quantum hardware (2024), \url{https://arxiv.org/abs/2406.19876}

\bibitem{Romero_2019}
Romero, J., Babbush, R., McClean, J.R., Hempel, C., Love, P.J., Aspuru-Guzik,
  A.: Strategies for quantum computing molecular energies using the unitary
  coupled cluster ansatz. Quantum Science and Technology  \textbf{4}(1),
  014008 (oct 2018). \doi{10.1088/2058-9565/aad3e4},
  \url{https://dx.doi.org/10.1088/2058-9565/aad3e4}

\bibitem{seidl2012compiler}
Seidl, H., Wilhelm, R., Hack, S.: Compiler Design: Analysis and Transformation.
  Springer (2012)

\bibitem{doi:10.1137/S0097539795293172}
Shor, P.W.: Polynomial-time algorithms for prime factorization and discrete
  logarithms on a quantum computer. SIAM Journal on Computing  \textbf{26}(5),
  1484--1509 (1997). \doi{10.1137/S0097539795293172},
  \url{https://doi.org/10.1137/S0097539795293172}

\bibitem{10313756}
Veshchezerova, M., Somov, M., Bertsche, D., Limmer, S., Schmitt, S.,
  Perelshtein, M., Joshi~Tripathi, A.: { A Hybrid Quantum-Classical Approach to
  the Electric Mobility Problem }. In: 2023 IEEE International Conference on
  Quantum Computing and Engineering (QCE). pp. 636--641. IEEE Computer Society,
  Los Alamitos, CA, USA (Sep 2023). \doi{10.1109/QCE57702.2023.00078},
  \url{https://doi.ieeecomputersociety.org/10.1109/QCE57702.2023.00078}

\bibitem{10189136}
Xia, S., Zhao, J.: Static entanglement analysis of quantum programs. In: 2023
  IEEE/ACM 4th International Workshop on Quantum Software Engineering (Q-SE).
  pp. 42--49 (2023). \doi{10.1109/Q-SE59154.2023.00013}

\bibitem{zhao2023qchecker}
Zhao, P., Wu, X., Li, Z., Zhao, J.: Qchecker: Detecting bugs in quantum
  programs via static analysis. In: 2023 IEEE/ACM 4th International Workshop on
  Quantum Software Engineering (Q-SE). pp. 50--57. IEEE (2023)

\end{thebibliography}

%




\end{document}